\documentclass[a4paper]{amsart}
\pdfoutput=1

\usepackage{todonotes}
\usepackage{algorithm}
\usepackage{algpseudocode}

\algblock{ParFor}{EndParFor}
\algnewcommand\algorithmicparfor{\textbf{parallel for}}
\algnewcommand\algorithmicpardo{\textbf{do}}
\algnewcommand\algorithmicendparfor{\textbf{end\ parallel for}}
\algrenewtext{ParFor}[1]{\algorithmicparfor\ #1\ \algorithmicpardo}
\algrenewtext{EndParFor}{\algorithmicendparfor}

\usepackage{amsmath}
\usepackage{stmaryrd}
\DeclareMathOperator{\nnz}{nnz}

\usepackage{amsthm}
\newtheorem{theorem}{Theorem}

\usepackage[style=alphabetic]{biblatex}
\usepackage[foot]{amsaddr}

\usepackage[hidelinks]{hyperref}

\usepackage{pdflscape}

\title[Engineering Compressed Matrix Multiplication with the FWHT]{Engineering Compressed Matrix Multiplication with the Fast Walsh-Hadamard Transform}

\begin{document}

\author{Joel Andersson}
\email{joelande@chalmers.se}
\author{Matti Karppa}
\email{matti.karppa@gu.se}
\address{Chalmers University of Technology / University of Gothenburg}

\maketitle

\begin{abstract}

We present an implementation of Pagh's compressed matrix multiplication algorithm, a randomized algorithm that constructs
sketches of matrices to compute an unbiased estimate of their product. By leveraging fast polynomial multiplication 
via the FFT, the algorithm achieves high performance when the product matrix is sparse or contains only a small number
of entries with magnitudes significantly larger than the rest. We show empirically that the algorithm is practical 
and can outperform state-of-the-art \texttt{DGEMM} implementations when the product matrix has few nonzero entries or is otherwise dominated by a small subset of elements with large magnitude. 
As a minor theoretical contribution, we replace the FFT with the Fast Walsh-Hadamard Transform (FWHT) in sketched 
multiplication, preserving all correctness and variance guarantees of the original algorithm.

Experiments with our carefully engineered multithreaded CPU implementation for dense double-precision matrices on
64-core CPU nodes across a range of synthetic benchmarks, exhibiting variable sparsity patterns,
show that the FWHT variant is up to 4 times faster than the FFT-based version. 
Under favorable sparsity and magnitude patterns in the product matrix, 
our FWHT-based implementation achieves a speedup of up to 40 over \texttt{DGEMM} from Intel MKL, 
with low probability of error in the estimates.
Our implementation is released as free software and comes with NumPy-compatible Python bindings.
\end{abstract}

\newpage

\section{Introduction}
\label{sect:introduction}

Matrix Multiplication (MM) is one of the most fundamental primitive operations in modern computation, 
underlying such ubiquitous applications as scientific computation~\cite{GolubL:2013,TrefethenB:1997} and 
modern machine learning~\cite{GoodfellowBC:2016}, so much so that the Generalized Matrix Multiply \texttt{DGEMM} routine from BLAS~\cite{BlackfordDDDHHHKLPPRW:2002}
is routinely used to evaluate world's top supercomputers~\cite{DongarraLP:2003,top500:2025}.
The elementary MM algorithm has seen intensive engineering effort, resulting in very powerful and highly optimized implementations, such
as in Intel Math Kernel Library (MKL)~\cite{Intel:2025}.

In some applications, such as detecting correlations~\cite{KarppaKK:2018,KarppaKKC:2020}, 
we do not need to compute the full product but only want to reliably find \emph{heavy hitters},
the elements which have a larger magnitude than the rest. 
among the sparse product of dense matrices. In 2012, Pagh~\cite{Pagh:2013} presented an interesting randomized algorithm
that addresses precisely this problem. In a nutshell, the idea is to \emph{compress} the operand matrices into polynomials
using 2-wise independent families of hash functions, use Fast Fourier Transform (FFT) to perform fast polynomial multiplication,
and then \emph{extract} an unbiased estimate for the elements of the product matrix from the polynomial product. 
The resulting estimator has
a variance bounded by the Frobenius norm of the product matrix, meaning the algorithm is \emph{output sensitive} and works best
when the product is sparse, that is, dominated by few large elements in absolute value.
We show empirically that Pagh's algorithm is not just a theoretical curiosity, but can be practical and outperform state-of-the-art
\texttt{DGEMM} under favorable circumstances. As a minor theoretical contribution, we show that FFT can be replaced by the
Fast Walsh-Hadamard Transform (FWHT), preserving all theoretical guarantees of the original algorithm.
Specifically, we
\begin{itemize}
    \item Show that FFT can be replaced with FWHT without affecting the theoretical guarantees of the algorithm,
    \item Provide a multithreaded CPU implementation of the algorithm as free software, implemented in C++ with NumPy compatible Python bindings,
    \item Empirically evaluate the implementation on various synthetic datasets, exhibiting differing sparsity patterns, 
    to show that the FWHT-based implementation can be up to 4 times faster than the FFT-based algorithm and up to 40 times faster
    than \texttt{DGEMM}, assuming favorable circumstances.
\end{itemize}

\subsection{Related work}

Since Strassen's 1969 result~\cite{Strassen:1969}, a line of theoretical work 
\cite{Pan:1982,Strassen:1987,CoppersmithW:1990,Williams:2012,LeGall:2014} has sought the bound on time $\mathcal O(n^{\omega})$ that $n\times n$ matrices can be multiplied, 
currently known to be $\omega<2.37286$, due to Alman and Williams~\cite{AlmanW:2024}.
Another line of research has applied group-theoretic and combinatorial methods~\cite{CohnU:2003,CohnKSU:2005,CohnU:2013}.
Despite impressive theory, only Strassen-like algorithms have seen substantial engineering effort~\cite{BallardDHLS:2012,BensonB:2015,HuangSHG:2016,HuangRMG:2017,HuangYG:2018,KarppaK:2019,HuangYG:2020,SchwartzV:2023},
as more complex algorithms are widely believed to be impractical. In fact, most state-of-the-art libraries
rely on engineering the elementary algorithm when implementing their GEMM routine~\cite{GotoG:2008,SmithGSHZ:2014,ZeeG:2015,KerrMDT:2017,XuZG:2023},
largely explained by the predictable and linear memory access patterns of the
elementary algorithm, conforming to the requirements of modern hardware, especially
in terms of cache use.
There has also been a lot of work in the theory of \emph{Sparse matrix multiplication}~\cite{Gustavson:1978,YusterZ:2005,AmossenP:2009,BallardBDGLST:2013,CampagnaKP:2013,PaghS:2014,AbboudBFK:2024}, as well as 
engineering of the equivalent primitive \texttt{SpGEMM}~\cite{BulucG:2012,BorstnikVWH:2014,LiuV:2014,NagasakaNM:2017,DeveciTR:2018,Davis:2019,ZhaoLWZWXLLPLZW:2025}. See also~\cite{GaoJCHWLW:2023} for a survey.

Pagh's algorithm makes use of \emph{probabilistic sketching}, constructing smaller data structures from larger ones 
that enable approximated queries on expensive statistics about the larger structure with probabilistic guarantees,
building on the AMS sketch~\cite{AlonMS:1999}. Random hash functions are a standard tool
in sketching, such as in set membership or cardinality estimation~\cite{Bloom:1970,FlajoletFGM:2007,HeuleNH:2013,XiaoCZL:2020,WangP:2023,Ertl:2024} and Locality-Sensitive Hashing (LSH)~\cite{AndoniI:2006,AndoniR:2015,IndykM:1998}. A curious connection to matrix multiplication may be found in the fact
that the structural tensor encoding matrix multiplication is known to have a \emph{probabilistic rank deficiency}, enabling
faster randomized algorithms~\cite{KarppaK:2019probabilistic}.

\subsection{Paper organization}
The remainder of this paper is organized as follows. In Section~\ref{sect:preliminaries}, we present the mathematical preliminaries and notation. In Section~\ref{sect:algorithm}, we recap Pagh's algorithm and modify it by replacing the FFT with the FWHT.  In Section~\ref{sect:implementation}, we present our implementation and the engineering considerations.
In Section~\ref{sect:experiments}, we describe the experimental setup, including the construction of the synthetic datasets. In Section~\ref{sect:results}, we describe the results
of our experiments, including scalability and accuracy of results. Section~\ref{sect:conclusion} concludes the paper.

\section{Preliminaries}
\label{sect:preliminaries}

We write $[n] = \{1,2,\ldots,n\}$. 
We denote the zero vector of $n$ elements by $\mathbf{0}_n$. We denote the Hadamard (elementwise) product of vectors by $x\circ y$. For a matrix~$A$, we denote its $i^\textrm{th}$ column and row by $a^\top_i$ and $a_i$, respectively, that is, we implicitly assume vectors are row vectors.
We denote the outer product of two vectors by $a\otimes b$.
Unless said otherwise, we are going to assume all matrices are square~$n\times n$ matrices where $n$ is a power of two, even though 
Pagh's algorithm~\cite{Pagh:2013} and our implementation extend to rectangular cases as well. We denote the Frobenius norm
by $\left\Vert A \right\Vert_F = \sqrt{\sum_{i,j=1}^n a_{ij}^2}$. We denote by $\nnz A=|\{ (i,j)\in[n]\times [n] \mid a_{ij}\neq 0 \}|$ the
number of non-zeros in~$A$.

We denote the discrete Fourier transform (FFT) and its inverse (IFFT) of a sequence $x$ by $\mathcal F \{ x\}$ and $\mathcal F^{-1}\{ x \}$, respectively. Likewise, we denote the Walsh-Hadamard transform of a sequence $x$ by $\mathcal H\{x\}$.
We denote the \emph{cyclic convolution} of two $n$-element sequences $x$ and $y$ by
\[
(x * y)[k]=\sum_{i+j\equiv k\mod n} x[i]y[j] \, .
\]
Note that this same sequence is often written $(x * y)[i] = \sum_{j=0}^{n-1} x[j]y[i-j]$, where the sequences are implicitly
extended periodically into infinity such that $y[i]=y[i\mod n]$.

We denote the bitwise XOR of two $n$-bit numbers by $i\oplus j$. We denote the \emph{XOR convolution} of two $n$-element sequences of
$x$ and $y$ by
\begin{equation}
\label{eq:fwhtxorconvolution}
(x *_{\oplus} y)[k] = \sum_{i\oplus j = k} x[i]y[j] \, .
\end{equation}
Note that the same sequence is often written $(x *_{\oplus} y)[i]=\sum_{j=0}^{n-1} x[j]y[i\oplus j]$. See, for example,~\cite[Chapter~23.8]{Arndt:2011}.

We say that a family (set) of hash functions $H$ from a domain $U$ to a set of buckets $[n]$ is \emph{2-wise independent} 
if all $x\in U$ are uniformly
hashed into $[n]$, and each \emph{pair} of hashing events is independent.
That is, $\Pr_{h\sim H}[h(x)=a]=1/n$ for all $x\in U$ and $a\in [n]$, and 
it holds for all
$x\neq y \in U$ and all (not necessarily distinct) $a,b\in [n]$
that 
\[
\Pr_{h\sim H} [h(x)=a \textrm{ and } h(y)=b] = \Pr_{h\sim H}[h(x)=a] \Pr_{h\sim H}[h(x)=b] \, .
\]

\section{Algorithm}
\label{sect:algorithm}

\subsection{Pagh's algorithm recap}

We recap Pagh's algorithm~\cite[Fig.~1]{Pagh:2013} here for completeness. Pagh's algorithm is a randomized algorithm
that consists of two subroutines: \textsc{Compress} that constructs a sketch of the product $C=AB$ and \textsc{Decompress}
that extracts an unbiased estimate of an element~$c_{ij}$ from the sketch. 
The underlying idea is based on the observation that the matrix product can be written as a sum of vector outer products:
\[
    C = \sum_{k=1}^n a^\top_k \otimes b_k = AB \, .
\]
We can then sketch each such outer product separately by constructing polynomials $p_A,p_B : [b]\to \mathbb R$ for 
$A$ and $B$, respectively, 
using random hash functions $h_1,h_2:[n]\to[b]$ to map each row of $A$ and each column of $B$ to a random
coefficient. We accumulate the values into the coefficients of the polynomial, using
another two random hash functions $s_1,s_2:[n]\to\{-1,1\}$ to provide a random sign to the each element.
That is, for a fixed $k\in[n]$, for each $i\in [n]$, we accumulate $p_A[h_1(i)]\gets p_A[h_1(i)]+s_1(i)A_{ik}$,
and correspondingly for each $j\in[n]$ we accumulate $p_B[h_2(j)]\gets p_B[h_2(j)]+s_2(j)B_{kj}$.

We then use FFT to transform the coefficients into a point-value representation,
meaning the polynomial sketch for the product can be obtained by an elementwise multiplication,
followed by IFFT. This amounts to computing the cyclic convolution of the coefficient sequences 
$p = p_A * p_B$. Specifically, choosing the hash functions from a 2-wise independent family means the 
hash functions can be decomposed modulo $b$: $h(i,j)=h_1(i)+h_2(j)\mod b$ and $s(i,j)=s_1(i)s_2(j)$, 
preserving 2-wise independence~\cite{CarterW:1979}. Thus, the procedure accumulates the elements of $c_{ij}$
into the coefficient $p[h(i,j)]$. As usually $b\ll n^2$, this means that multiple elements can be accumulated
in the same coefficient, but the noise elements should cancel in expectation, due to 
random signs. Finally, the procedure is repeated independently~$d$ times and the median of the estimates is used as output.

In addition to the choice of the hash function families, the algorithm has two free parameters: the size of the sketch~$b$ and 
the number of sketches~$d$. Pagh proves that the estimate $\tilde c_{ij}$ is unbiased and has a variance
bounded by $\left\Vert AB \right\Vert_F^2/b$~\cite[Theorem~3.1]{Pagh:2013}, and that \textsc{Compress} runs in time $\mathcal O(d(\nnz A + \nnz B + n b \log b))$~\cite[Lemma~2.2]{Pagh:2013}. A consequence of this is that if the input matrices are dense and square, 
then the running time becomes $\mathcal O(dn^2+dn b\log b)$. The $b\log b$ term is a direct consequence of applying Cooley-Tukey 
FFT~\cite{CooleyT:1965} on the polynomial of length~$b$, and $dn$ comes from the number of outer product sketches. 
The variance bound makes the algorithm \emph{output sensitive}:
the algorithm is best suited to cases where the output is sparse and contains only a few \emph{heavy hitters}. That is,
most elements are either zero or close to zero, and the heavy hitters with a significantly larger magnitude
dominate the Frobenius norm.

\subsection{Replacing the FFT with the FWHT}

Heavy lifting in the algorithm
is done by the FFT and the IFFT. However, the FFT can be easily
replaced with the Fast Walsh-Hadamard Transform (FWHT), as shown in Algorithm~\ref{algo:paghfwht}.
It is well known that the FWHT can be computed in a manner analogous to the Cooley-Tukey FFT (see for example~\cite[Chapter~23]{Arndt:2011}) in $\mathcal O(n\log n)$ time, so this change does not affect asymptotics of the algorithm.

However, the FWHT has other important properties that improve the running time if the running time is compute-bound.
Firstly, computing the FWHT does not require complex arithmetic, as all coefficients involved are $\pm 1$,
apart from the (potentially irrational and in our case irrelevant) normalization factor. This also means
that the computation does not need to make use of multiplications. Furthermore, $\mathcal H\{\mathcal H\{x\}\}=x$, so
the routine is its own inverse, and the transform can be easily computed in-place. This means that the FWHT has
a smaller memory footprint, making it somewhat cache-friendlier.

Another difference is that $\mathcal H\{\mathcal H\{x\}\circ\mathcal H\{y\}\}[k] = (x *_{\oplus} y)[k] = \sum_{i\oplus j=k}x[i]y[j]$, as stated in~\eqref{eq:fwhtxorconvolution}. This has the effect that 
if, in the original algorithm, the monomials corresponding to 
$c_{ij}$ were accumulated in the coefficient $p[h_1(i)+h_2(j)\mod b]$, in this case they are accumulated in the coefficient
$p[h_1(i)\oplus h_2(j)]$. However, it is obvious that the number of monomials accumulated is still equal. Furthermore, it is well
known that $h(i,j)=h_1(i)\oplus h_2(j)$ is 2-wise independent if $h_1$ and $h_2$ are independently sampled from a 2-wise independent 
family~\cite{CarterW:1979}, so it is obvious that Algorithm~\ref{algo:paghfwht} preserves all the properties of the original algorithm. Although the result is rather trivial, we phrase this as a theorem.

\begin{algorithm}[t]
\begin{algorithmic}[1]
    \Function{Compress}{$A,B,d,b$}
    \For{$t\in[d]$}
        \State Draw $s_1[t],s_2[t] : [n]\to \{-1,+1\}$ from a 2-wise independent hash family
        \State Draw $h_1[t],h_2[t] : [n]\to [b]$ from a 2-wise independent hash family
        \State $p[t]\gets \mathbf{0}_b$ \Comment{Real-valued}
        \For{$k\in[n]$}
            \State Let $p_A\gets \mathbf 0_n$
            \For{$i\in[n]$}
                \State $p_A[h_1[t](i)]\gets p_A[h_1[t](i)] + s_1[t](i)A_{ik}$
            \EndFor
            \State Let $p_B\gets \mathbf 0_n$
            \For{$j\in[n]$}
                \State $p_B[h_2[t](j)]\gets p_B[h_2[t](j)] + s_2[t](j)B_{kj}$
            \EndFor
            \State $p[t]\gets p[t]+\mathcal H\{p_A\} \circ \mathcal H\{p_B\}$
        \EndFor
        \State $p[t]\gets \mathcal H\{p[t]\}$
    \EndFor
    \State \Return ($p,s_1,s_2,h_1,h_2$)
    \EndFunction
    \State
    \Function{Decompress}{$p,s_1,s_2,h_1,h_2,i,j$}
        \State $x\gets \mathbf 0_d$
        \For{$t\in[d]$}
            \State $x_t\gets s_1[t](i)s_2[t](j)p[t][h_1[t](i) \oplus h_2[t](j)]$
        \EndFor
        \State \Return \Call{Median}{$x$}
        \EndFunction
\end{algorithmic}
\caption{The FWHT version of Pagh's algorithm. The main difference is that the transform can be performed in-place and modular addition has been changed to XOR. In \textsc{Compress}, the hash functions $s_1,s_2,h_1,h_2$ are stored in an array and returned together with the product polynomial $p$, constituting the sketch together. The function \textsc{Decompress} returns the element $\tilde c_{ij}$ of the sketched product. Note that the FWHT is its
own inverse: $\mathcal H^{-1}\{x\}=\mathcal H\{x\}$.}
\label{algo:paghfwht}
\end{algorithm}

\begin{theorem}
Algorithm~\ref{algo:paghfwht} satisfies all the properties of Pagh's algorithm~\cite{Pagh:2013}. Consequently, Lemma~2.2 and Theorems 3.1--3.5 of~\cite{Pagh:2013}
also hold for Algorithm~\ref{algo:paghfwht}.
\end{theorem}
\begin{proof}
    Both algorithm have the same \emph{transform $\rightarrow$ elementwise multiplication $\rightarrow$ inverse transform} structure to compute the convolution
    of sequences of length~$b$: cyclic convolution for FFT and XOR-convolution for FWHT. Exactly $b$ product monomials from the constituent polynomials are accumulated in each coefficient of the polynomial~$p$, and each  $a_{ik}b_{kj}$ contributes to exactly one polynomial coefficient, so the sets of monomials are disjoint across coefficients.
    Moreover, since the hash functions $h_1,h_2,s_1,s_2$ are drawn from 2-wise independent hash families, the combined hash function is
    2-wise independent in both cases, inducing a similar partitioning of the monomials into $b$ buckets with the same distributional guarantees
    as in the proofs of~\cite{Pagh:2013}. Consequently, the error and variance bounds carry over unchanged.
    Finally, FWHT runs in $\mathcal O(b\log b)$ time by a Cooley-Tukey construction, preserving the asymptotic running times.
\end{proof}

\section{Implementation}
\label{sect:implementation}

\subsection{Overview}

We have implemented the algorithm in C++, optimized for Intel CPUs. The implementation is available as free software at
our GitHub repository\footnote{\url{https://github.com/mkarppa/compmatmul}}. 
The implementation uses FFTW for computing the FFT~\cite{FrigoJ:2005} and FXT for computing the FWHT~\cite{Arndt:2025}.
For linear algebra and array routines, we use the Intel MKL~\cite{Intel:2025}. We provide 
NumPy-compatible Python bindings with pybind11~\cite{JakobRM:2017}. 
We provide an Apptainer definition file for containerization in standard HPC environments. 
For reference baseline, we provide a wrapper with corresponding interface to the MKL implementation of \texttt{DGEMM},
the standard BLAS matrix multiplication primitive.
The present implementation described here is 
limited to only double-precision floating-point numbers due to numerical stability issues with lower precision floats
which we hope to address in the full version of this paper. The current implementation also only supports dense
input and output matrices. Although we have not reported experiments in this conference version, the implementation also supports
rectangular matrices, in addition to square. The structure of the FWHT implementation of \textsc{Compress}
is shown in Algorithm~\ref{algo:implementationfwht}.

\begin{algorithm}[t]
\begin{algorithmic}[1]
    \Function{Compress}{$A,B,d,b,h_1,h_2,s_1,s_2$}
    \State Allocate output array $p$ of $d\times b$ real words
    \State Initialize $d$ locks
    \State For each thread, allocate local real array $p_A$ of size $d\times b$, and $p_B$ of length $b$
    \ParFor{$k\in [n]$}
        \For{$t\in[d]$}
            \State $p_A[t]\gets \mathbf 0_b$
            \For{$i\in [n]$}
                \State $p_A[t][h_1[t](i)] \gets p_A[t][h_1[t](i)] + s_1[t](i)a^\top_{ki}$
            \EndFor
            \State $p_A[t] \gets \mathcal H\{p_A[t]\}$
        \EndFor
        \For{$t\in[d]$}
            \State $p_B \gets \mathbf 0_b$
            \For{$j\in [n]$}
                \State $p_B[h_2[t](j)] \gets p_B[h_2[t](j)] + s_2[t](j)b_{kj}$
            \EndFor
            \State $p_B \gets \mathcal H\{p_B\}$
            \State $p_B \gets p_A[t]\circ p_B$
            \State Acquire lock $t$
            \State $p[t]\gets p[t] + p_B$
            \State Release lock $t$
        \EndFor
    \EndParFor
    \State For each thread, deallocate $p_A,p_B$
    \ParFor{$t\in[d]$}
        \State $p[t]\gets \mathcal H\{p[t]\}$
    \EndParFor
    \State \Return $p$
    \EndFunction
\end{algorithmic}
\caption{Implementation of the FWHT version of the compression routine. Note that the order of loops
is changed, the left-hand-operand is transformed first, and the transform of the right-hand
operand is interleaved with product accumulation.}
\label{algo:implementationfwht}
\end{algorithm}

\subsection{Hash functions}

The analysis in~\cite{Pagh:2013} does not specify the hash function families, other than requiring them to be 2-wise independent.
In our experiments, we use the 32-bit \emph{multiply-add-shift} hash functions, that is,
the hash function family is $H = \{ h_{a,c} : [2^{32}]\to[b]\mid a,c\in [2^{64}]\}$ where 
$h_{a,c}(x)=\left\lfloor\frac{ax+c}{2^{64-\log_2 b}}\right\rfloor$, that is, we perform a 64-bit multiplication and addition, and
take the $\log_2 b$ highest bits of the result as the hash value. The hash function is constructed simply by drawing 
two random 64-bit integers $a$ and $c$. This family of hash functions is well known to be 2-wise independent~\cite{CarterW:1979} and 
practical~\cite{Thorup:2015}. We use the same family for both the bucket hashes $h_1,h_2$
and the sign hashes $s_1,s_2$, where the hash is obtained by simply setting $b=2$.

We have also implemented and experimented with simpler (multiply-add) and more complex (simple tabulation hashing~\cite{PatrascuT:2012}) hash functions, but these experiments have been omitted from this conference version.

\subsection{Memory access patterns}
As the algorithm makes rather intensive use of memory accesses, we have applied some care in ensuring
mostly sequential access in memory, and tried to avoid filling up the CPU cache prematurely.
Assuming the input matrices are stored in row-major order, considering that we are trying to sketch
outer products of the form $a_k^\top b_k$, it makes sense to have the left-hand operand transposed.
This is contrary to the ordinary wisdom of having the right-hand operand transposed for ordinary matrix multiplication.

The most memory-access-intensive part of the implementation is in the compression routine when the operand-specific
polynomial sketches are constructed. To this end, our implementation inverses the order of loops from the abstract
description of 
Algorithm~\ref{algo:paghfwht}, and we construct the $dk$ untransformed polynomial
pairs simultaneously, with the outer loop iterating over the inner dimension and the inner loop iterating over the independent
sketches. This order is friendlier to the cache, as we then better preserve locality of reference. 

In the case of the FFT, After the $d$ independent sketches
have been constructed, we transform them simultaneously in parallel. Finally, we perform elementwise multiplication
and inverse transforms in parallel. 

The FWHT version is somewhat simpler as it requires much less auxiliary space. We make fewer allocations and transform all 
left-hand side polynomials first. Then, we perform the right-hand side transforms one-by-one, interleaved with computing
and accumulating the products, reusing the memory. All transforms
are performed in-place.

\subsection{Multithreading}

We use OpenMP for multithreading. In both the FFT and FWHT cases, we have organized our memory accesses such 
that the only place where read-after-write can
occur is in the inner loop of the sketch construction, after the transform and elementwise multiplication.
That is consequently the only place where locks are necessary.

The FFT and the FWHT versions have been organized slightly differently, but both have two distinct parallel blocks,
one for constructing the sketches of the operands and accumulating the transformed polynomials, 
and a second one for performing the inverse transforms. The difference is that, in the FFT version,
there are two separate inner for loops, as the sketches are constructed for both $p_A$ and $p_B$ completely before performing the
transforms. 
In the FWHT case, only $p_A$ is constructed first, and the transforms are interleaved with the sketch construction
of $p_B$. This is enabled by the fact that we can do a lot of the operations in-place, with better locality of reference.

\section{Experiments}
\label{sect:experiments}

\subsection{Overview}

We have conducted experiments to empirically evaluate the quality of the implementation with respect to the theoretical guarantees
of the algorithm, and to also empirically evaluate the running time. All experiments have been conducted with synthetic
data that has been crafted with instances of varying difficulty. 
The experiments were run on two different
HPC clusters with two kinds of nodes, with 64 CPU cores each, called \emph{Vera} and \emph{Minerva}.
The hardware information is reported in Table~\ref{tab:hardware}.
All experiments were run in Apptainer containers, derived from Ubuntu Linux, using the Miniforge3 image.
The experiments were executed using SLURM. All experiments were run using the Python bindings, that is, the experiment scripts were written
in Python.

Due to the setup of the clusters, even though exclusive access to the nodes was
requested, there appeared to be a lot of difficult to control sources of noise, resulting in unexplainable variations
in running times due to factors that we were not able to control ourselves. Timing results from Minerva are more reliable,
as we had direct BIOS and sudo access to the machine and could disable some sources of noise, such as CPU frequency scaling
and boost clocks.

\begin{table}[t]
    \centering
    \caption{Hardware configuration of the nodes the experiments were run on.}
    \label{tab:hardware}
    \begin{tabular}{lrr}
        \textbf{Property}     & \textbf{Vera}          & \textbf{Minerva}       \\\hline
        CPU model             & Intel Xeon Platinum 8358 & Intel Xeon Gold 6548N \\
        Sockets$\times$Cores  & $2\times 32$             & $2\times 32$ \\
        Max frequency         & 3.4~GHz                  & 4.1~GHz\\
        Base frequency        & 2.6~GHz                  & 2.8~GHz\\
        Cache per socket      & 48~MiB                   & 60~MiB\\
        Architecture          & Ice Lake                 & Emerald Rapids\\
        Total system RAM      & 1024~GiB                 & 512~GiB \\
    \end{tabular}
\end{table}

\subsection{Instances}

All instances consisted of dense input matrices,
and product matrices had either one, $\log_2 n$, or $n$ \emph{big} elements. The remaining elements are called \emph{small}.
The task is to be able to robustly identify the big elements from the small in the output, preferably with low absolute error.
Depending on the instance, the small elements may be all zeros, mostly zeros or near-zeros, or otherwise significantly smaller elements
than the big elements. This means that some instance outputs can be considered sparse while others are dense. We provide a more
detailed description of each instance below. Table~\ref{tab:instances} also summarizes the instances.

\begin{table}[t]
    \centering
    \caption{Summary of the different instance types. The properties relate to the product matrices.}
    \label{tab:instances}
    \begin{tabular}{@{}l@{\,}l@{\,}r@{\,\,}r@{\,\,}r@{}}
        \textbf{Instance}          & \textbf{Big value} & \textbf{\#Big elems} & \textbf{\#Non-zeros}   & \textbf{Sparse/Dense} \\\hline
         \textsc{Logunit}          & 1.0 & $\log_2 n$               & $n$                     & Sparse                \\
         \textsc{Diagonal}         & $0.5\leq |c_{ij}|\leq 1.0$  & $n$                      & $n$                     & Sparse                \\
         \textsc{Covariance}       & $\approx 0.8$ & 1                        & $n^2$                   & Dense                 \\
         \textsc{Lightbulb}        & $\approx 0.8$ & 1                        & $n^2$                   & Dense                 \\
    \end{tabular}
\end{table}

\textbf{Logunit.} 
The easiest instance type is called \emph{logunit} for ``logarithmic number of units (ones)''.
The output matrix $C$ contains exactly $\log_2 n$ elements with value 1, constituting the big elements,
together with $n-\log_2 n$ elements with value 0.001 which are considered essentially 0 and small.
The values have been placed such that each column and row contain exactly 1 non-zero. Thus, the matrix has full rank.
The input matrices have been constructed as follows: The left-hand operand matrix $A=\tilde AD_1$,
where $\tilde A$ has been drawn uniformly at random from $U(-1,1)$, and $D_1$ is a diagonal matrix containing
randomly assigned $\left\lfloor \frac{\log_2 n}{2}\right\rfloor$ elements with value 100, with the remaining diagonal elements
set at 0.01. The right-hand operand $B=\tilde A^{-1}D_2P$  where $D_2$ is a similar diagonal matrix with random $\left\lceil \frac{\log_2 n}{2}\right\rceil$ elements, disjoint from the big elements of $D_1$, with value 100 and the remaining diagonal elements set to 0.01. $P$ is a random permutation matrix. The resulting
product $C=AB=D_1\tilde A\tilde A^{-1}D_2P=(D_1D_2)P=DP$ has exactly the desired properties. 

\textbf{Diagonal.}
A slightly more difficult variant of the logunit, we construct two matrices such that the resulting matrix is a diagonal matrix,
the columns (or rows) of which have been permuted randomly. The diagonal elements are drawn from the interval
$[-1,1)\setminus [-0.5,0.5)=[-1,-0.5)\cup[0.5,1)$, that is, all diagonal elements $x$ are required to satisfy $0.5\leq |x| \leq 1$.
This means that the product contains $n$ non-zeros, which are the big elements of the instance.
The construction is similar: we draw the random diagonal matrix $D$, the random permutation matrix $P$, and matrix $A$ drawn uniformly at random from $U(-1,1)$, and then set $B=A^{-1}PD$, thus $AB=AA^{-1}PD=PD$. 

\textbf{Covariance.}
This instance aims for finding a single strongly correlated pair of random variables across a set of two,
mostly uncorrelated Gaussian variables. There is thus only one big element, but $n^2$ non-zeros in the output.
The instance is constructed by first drawing two independently random Gaussian matrices $A$ and $\tilde B$.
Then, we plant a correlated pair $(i^*,j^*)$ by selecting a random row in $A$ and a random column in $B$,
and then construct $B$ such that 
\[
b^\top_j=
\left\{
\begin{array}{rl}
\rho a_{i^*} + \sqrt{1+\rho^2} \tilde b^\top_{j^*} & \textrm{if } j= j^* \, ,\\
\tilde b^\top_j & \textrm{otherwise.}
\end{array}
\right.
\]
We use $\rho = 0.8$. This results in a product matrix where there is exactly one element $c_{i^*j^*}\approx 0.8$, and
all other elements are close to zero, for large enough $n$. This can be seen as a \emph{dichromatic} covariance
estimation problem, where there are two sets of variables, and one hidden pair of variables is correlated, so we aim to find
a pattern among noise. This instance is considerably harder than logunit or diagonal.

\textbf{Lightbulb.}
Perhaps the hardest instance is based on the lightbulb problem \cite{Valiant:1988}: a set of lightbulbs
is observed turning on and off uniformly at random, apart from a planted pair, that is correlated.
Specifically, we address the bichromatic outlier correlations version of the problem~\cite{KarppaKK:2018,KarppaKKC:2020}:
given two sets of random variables $X=\{x_1,x_2,\ldots,x_n\},Y=\{y_1,y_2,\ldots,y_n\}$, 
uniformly and independently distributed over $\{-1,1\}$, apart from a planted pair $(x_{i^*},y_{j^*})$
that is correlated, that is, $\left<x_{i*},y_{j*}\right>=\rho$, find the unique planted pair.
The matrices $A,\tilde B$ thus correspond to a sample drawn from uniform distribution over $\{-1,1\}$,
and the right hand matrix $B$ is constructed such that all other columns $b^\top_j$ are equal to 
$\tilde b^\top_j$, except the column $b^\top_{j^*}$, that is constructed by copying the row $a_{i*}$
and then flipping the sign of $d=\frac{n(1-\rho)}{2}$ elements, corresponding to the Hamming distance
that yields an inner product of $\rho$ between the planted vectors. We use $\rho=0.8$. Thus, 
the resulting product contains exactly one element $c_{i^*j^*}\approx 0.8$, and all other elements are close to zero,
for sufficiently large $n$, meaning there is one unique big element and $n^2$ non-zeros in the output.

\subsection{Variance experiments}

Theorem~3.1 in~\cite{Pagh:2013} gives us the guarantee that, having constructed a sketch $\tilde C$ of the product $C=AB$ using $d=1$ with
the algorithm, decompressing an entry $(i,j)$ yields an unbiased estimate~$\tilde c_{ij}$ of $c_{ij}$ with variance bounded by $\left\Vert A \right\Vert_F^2/b
$, that is, the variance depends on the overall magnitude of elements in the correct product matrix. In order to ensure that our
implementation conforms to this theoretical prediction, we fix $n=1024$, the input matrices $A,B$, and $(i,j)$, 
and repeatedly construct the sketch with randomly drawn hash functions, varying $b$ between different trials. 
For each $b$, we conduct 1000 repetitions. We use the estimates of an individual big element to compute the sample variance.

\subsection{Correctness experiments}
\label{sect:correctnessexperiments}

To evaluate the \emph{correctness of estimates}, we collect information about the output estimates. Specifically, as all
\emph{big} elements ought to have a magnitude in excess of 0.5, how many such elements are satisfied, and, vice versa, how many
\emph{small} elements have a magnitude below 0.5. The idea is that this enables us to tell these elements apart.
We also collect information about whether the elements are within 0.1 in absolute error from the correct values. Specifically,
denoting the output estimate matrix by $\tilde C$ and the correct by $C$, the set of big element indices by $B\subseteq [n]\times [n]$ and the set of small element indices by $S\subseteq [n]\times [n]$ we collect
\begin{itemize}
    \item The number of big elements in excess of 0.5, that is, $B_{\geq 0.5}=|\{(i,j)\in B \mid |\tilde c_{ij}| \geq 0.5 \}|$,
    \item The number of small elements less than 0.5, that is, $S_{\leq 0.5}=|\{(i,j)\in S \mid  | \tilde c_{ij}| \leq 0.5 \}|$,
    \item The number of big elements within 0.1 in absolute error of the correct value $B_{\epsilon\leq0.1}=|\{  (i,j)\in B \mid |\tilde c_{ij} -c_{ij} |\leq 0.1 \}|$, and
    \item The number of small elements within 0.1 in absolute error of the correct value $S_{\epsilon\leq0.1}=|\{ (i,j)\in S \mid |\tilde c_{ij} -c_{ij} |\leq 0.1 \}|$.
\end{itemize}

In a smaller experiment, we fix $n=8192$ and compute the product on 100 distinct matrices, each time with different random hash functions.
In a larger experiment, we fix $n=32768$ and $n=65536$ and compute the product on 10 distinct matrices, with 10 repetitions using different
random hash functions.

\subsection{Scaling experiments}

To evaluate the scaling of the algorithm as the function of the size of input~$n$, we vary $n=16,32,64,\ldots,65536,131072$,
that is, powers of 2, and perform 1+5 repetitions on the same input matrix with the same hash functions. The running time
of the algorithm is not particularly sensitive to the data, apart from the potential memory effects that might be caused
by the access patterns induced by the hash functions, but with large inputs, we expect any such effects to be drowned in noise from
other sources. The first extra repetition is included to warm up the cache, so that the remaining repetitions are comparable.
We report the wall clock time for running the experiment.

\subsection{Parameter selection}
\label{sect:experimentsparameterselection}

Pagh does not prescribe exact choice of parameters in~\cite{Pagh:2013}. However, Theorems~3.4 and~3.5 in \cite{Pagh:2013} 
give some indications. Theorem 3.4 suggests choosing $b\geq 8 \nnz C$ and $d\geq 6 \log_2 n$ to ensure that the
matrix product can be computed exactly; however, this is rather
pessimistic and does not extend well to cases where $C$ is dense. Theorem~3.5 suggests choosing
$d\geq 6\log_2 n$ to get a high-probability bound on the absolute error that depends on~$b$.

Drawing inspiration on these theorems and their proofs, we parametrize our algorithm with two constants $c_d,c_b>0$,
such that we set $d=2\cdot\left\lfloor\frac{c_d\log_2 n}{2}\right\rfloor+1$, that is, the next odd number greater than
or equal to $c_d\log_2 n$, and $b=c_bn$.

We perform parameter selection in a two-phase grid search. First on $n=1024$ matrices, with 100 repetitions,
yielding candidate parameter pairs which we categorize as follows:
\begin{itemize}
    \item \emph{Perfect} if $B_{\epsilon\leq 0.1}+S_{\epsilon\leq 0.1}=n^2$ in all cases, that is, we always are close to the
    correct value in absolute terms,
    \item \emph{Good} if $B_{\epsilon\leq 0.1}=B$ and $|S_{\epsilon\leq 0.1}|\geq |S|-n$ in all cases, that is, we always correctly compute
    the big elements but allow some slack for the small elements,
    \item \emph{Decent} if $|B_{\epsilon\leq 0.1}|\geq \max\{|B|-\log_2 n,1\}$ and $|S_{\epsilon \leq 0.1}|\geq 0.99|S|$ in all cases, that is we now also allow a little slack for the big elements as well, and
    \item \emph{Satisfactory} if $|B_{\geq 0.5}|\geq 0.99|B|$ in 99\% of cases and $|S_{\leq 0.5}|\geq 0.99|S|$ in 99\% cases.
\end{itemize}

These parameter choices may be overly strict because they emphasize finding the big elements in almost all cases. Having identified the candidate parameters, 
we only keep Pareto-optimal pairs and evaluate them against $n=8192$ matrices with 1+3 repetitions. For each
problem and parameter category, we select the fastest pair in the category, 
measured in wall clock time. These parameters are used in larger scaling and correctness experiments.

\section{Results}
\label{sect:results}

\subsection{Variance experiment}

\begin{figure}[t]
    \centering
    \includegraphics[width=\linewidth]{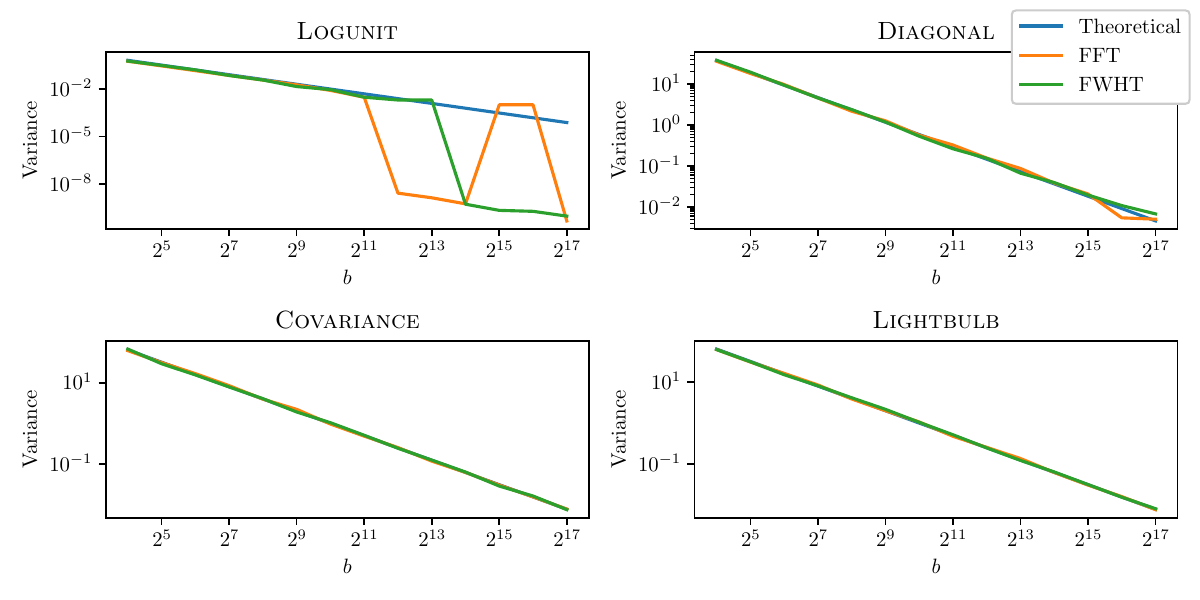}
    \caption{Variance experiment results. Note the logarithmic axes.}
    \label{fig:variancexperiment}
\end{figure}

Figure~\ref{fig:variancexperiment} shows the results of the variance experiment, for one random instance with $n=1024$ from
each instance categories, with 1000 independently drawn hash functions. Sample variance is shown as a function of the parameter $b$. We have fixed $d=1$. The blue line shows the theoretical bound of $\Vert C \Vert_F^2/b$. The variance bound is clearly satisfied rather tightly. The logunit instance stands out special because, with large $b$, the values are essentially correct in almost all cases, causing large deviations for individual outliers in logarithmic scale.

\subsection{Parameter selection}

Table~\ref{tab:parameters} lists the parameters selected according to the procedure of Section~\ref{sect:experimentsparameterselection}. The table lists the parameters $c_d$ and $c_b$ per problem, together with the 
parameter category. The relative difficulty of the different instances is obvious from the parameters, as very small parameters fall
in the \emph{perfect} category for \textsc{Logunit}, whereas \textsc{Covariance} or \textsc{Lightbulb}
would have required \emph{perfect} parameters beyond the scope of the grid search.

\begin{table}[t]
    \centering
    \caption{Parameters used for scaling experiments.}
    \label{tab:parameters}    
    \begin{tabular}{l|rr|rr|rr|rr}
    & \multicolumn{2}{c|}{\textbf{Perfect}} & \multicolumn{2}{c|}{\textbf{Good}} & 
    \multicolumn{2}{c|}{\textbf{Decent}} & \multicolumn{2}{c}{\textbf{Satisfactory}} \\
    \textbf{Problem}    & $c_d$ & $c_b$ & $c_d$ & $c_b$ & $c_d$ & $c_b$ & $c_d$ & $c_b$ \\\hline
    \textsc{Logunit}    & 1.0   & 0.5 & 0.25  & 1.0 & 0.25  & 0.25 & 0.25  & 0.25 \\
    \textsc{Diagonal}   & 3.25  & 4.0 & 2.0   & 4.0 & 1.75  & 2.0  & 0.75  & 4.0 \\
    \textsc{Covariance} & ---   & --- & ---   & --- & ---   & ---  & 1.5   & 4.0 \\
    \textsc{Lightbulb}  & ---   & --- & ---   & --- & ---   & ---  & 2.0   & 4.0 \\
    \end{tabular}
\end{table}

\subsection{Scaling experiment}

Figures~\ref{fig:verascaling} and~\ref{fig:minervascaling} show the scaling experiment results for Node~1 and~2, respectively. 
The marker shows the median of the 5 later runs. The errorbars show the minimum and maximum running times. All results are shown on a logarithmic scale. The results
show that both of our implementations beat \texttt{DGEMM} with all parametrizations on \textsc{Logunit}, and FWHT beats \texttt{DGEMM} up to \emph{Good}
parameters ($c_d=2.0,c_b=4.0$) on \textsc{Diagonal} where it becomes tied. The FWHT implementation narrowly beats 
\texttt{DGEMM} on \textsc{Lightbulb}.

The effect becomes even more pronounced if we look at the speedups in Figures~\ref{fig:veraspeedup} and~\ref{fig:minervaspeedup}. The upper row shows
the speedup of the FWHT routine vs. the \texttt{GEMM} baseline, whereas the lower row shows the speedup of the FWHT routine vs. the FFT
implementation. Speedups computed with both minimum and median values are shown. As expected, the speedup becomes more pronounced
with increased input size.
The speedup plot also shows prominent cache effects. 
There are speedup spikes at $n=2^6$, $n=2^9$, and $n=2^{16}$, suggesting bottlenecks with the 48~KiB L1 cache, 2~MiB L2 cache, and L3 cache, respectively.

For lack of space, we have only included the largest instances in Tables~\ref{tab:verascalingtabular} and~\ref{tab:minervascalingtabular} that show
the scaling of our implementation, together with speedups: SU1 = GEMM/FFT speedup, SU2 = GEMM/FWHT speedup, and SU3 = FFT/FWHT
speedup. It is not very surprising that the FFT vs. FWHT speedup  is intimately tied to the parameter~$b$ which is the size of the transform
computed; the larger the $b$, the more beneficial it is to use FWHT.

\subsection{Correctness experiment}
Table~\ref{tab:correctness} summarizess the results of the correctness experiment. The experiment was run with $n=65536$, and
10 distinct inputs were used. Each input was repeated 10 times with different hash functions. The table reveals that 
\textsc{Logunit} was very easy: even with worse parameters, we achieve perfect performance. \textsc{Diagonal} was likewise very easy,
and even satisfactory parameters provide near perfect performance.

However, \textsc{Covariance} and \textsc{Lightbulb} were considerably more difficult, as we have so much noise due to the dense
output matrix that, in terms of absolute error, we have only around 60\% hit rate with \emph{big} elements. However, even this parametrization is adequate to tell the \emph{big} and \emph{small} elements reliably apart: we can almost always identify
the unique big element without any ambiguity, with only one failure across all the 400 trials. The results are easier to interpret
if one keeps in mind that the estimate can be stated~\cite[Equation~2]{Pagh:2013}
\[
\tilde c_{i^*j^*}=c_{i^*j^*}+s(i^*,j^*)\sum_{\substack{h(i,j)=h(i^*,j^*) \\ i,j\neq i^*,j^*}}
s(i, j) c_{ij} \, ,
\]
 meaning that the estimate
is the correct value plus random-signed noise. As long as the noise does not accumulate too much absolute value,
the big element stands out.

\section{Conclusion}
\label{sect:conclusion}

We have presented a practical implementation of Pagh's compressed matrix
multiplication algorithm where we have replaced the FFT-based sketching with
one based on the FWHT, which preserves all theoretical guarantees and provides
clear performance advantages by being simpler to compute and making better
use of memory than FFT.

We have shown with a range of synthetic benchmarks that, although the algorithm
is not a general-purpose replacement to matrix multiplication, it can outperform
Intel's state-of-the-art CPU implementation of \texttt{DGEMM} under favorable circumstances,
up to a factor of 40. We also verified empirically that our FWHT implementation improves performance by as much as a factor of 4, satisfies the theoretical guarantees on estimator quality, and that the products computed by our algorithm are, in fact,
highly accurate.

\subsection*{Future work}
We have a lot of future work left for the full version of this
paper and future publications, including a more sophisticated parameter
selection strategy
and expanding on the empirical evaluation, determining the effect of different hash functions and the 
parallelization of the implementation, and more careful mapping of the speedup as a function of parameter choices.
The implementation itself could be extended: we need to add robust support for
single-precision floats, support for sparse matrix formats, and provide a GPU-based
implementation.

\section*{Acknowledgements}
Parts of the experiments were enabled by the computational and data storage resources 
provided by Chalmers e-Commons at Chalmers. Parts of the experiments were conducted using
the Minerva cluster at Chalmers Department of Computer Science and Engineering.

\section*{References}

\printbibliography[heading=none]

\newpage

\appendix

\section*{Appendix}

\begin{figure}[h!]
    \centering
    \includegraphics[width=0.95\linewidth]{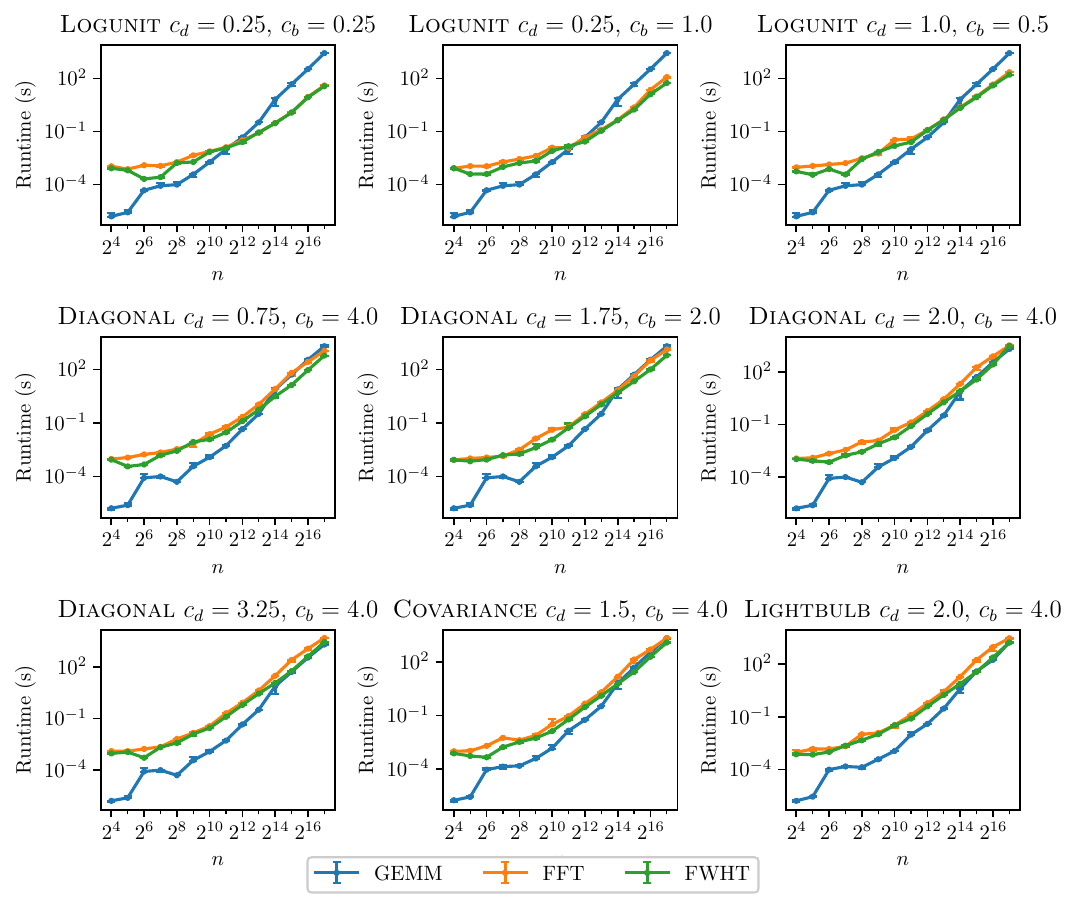}
    \caption{Scaling experiment results from Vera.}
    \label{fig:verascaling}
\end{figure}

\begin{figure}[h!]
    \centering
    \includegraphics[width=0.95\linewidth]{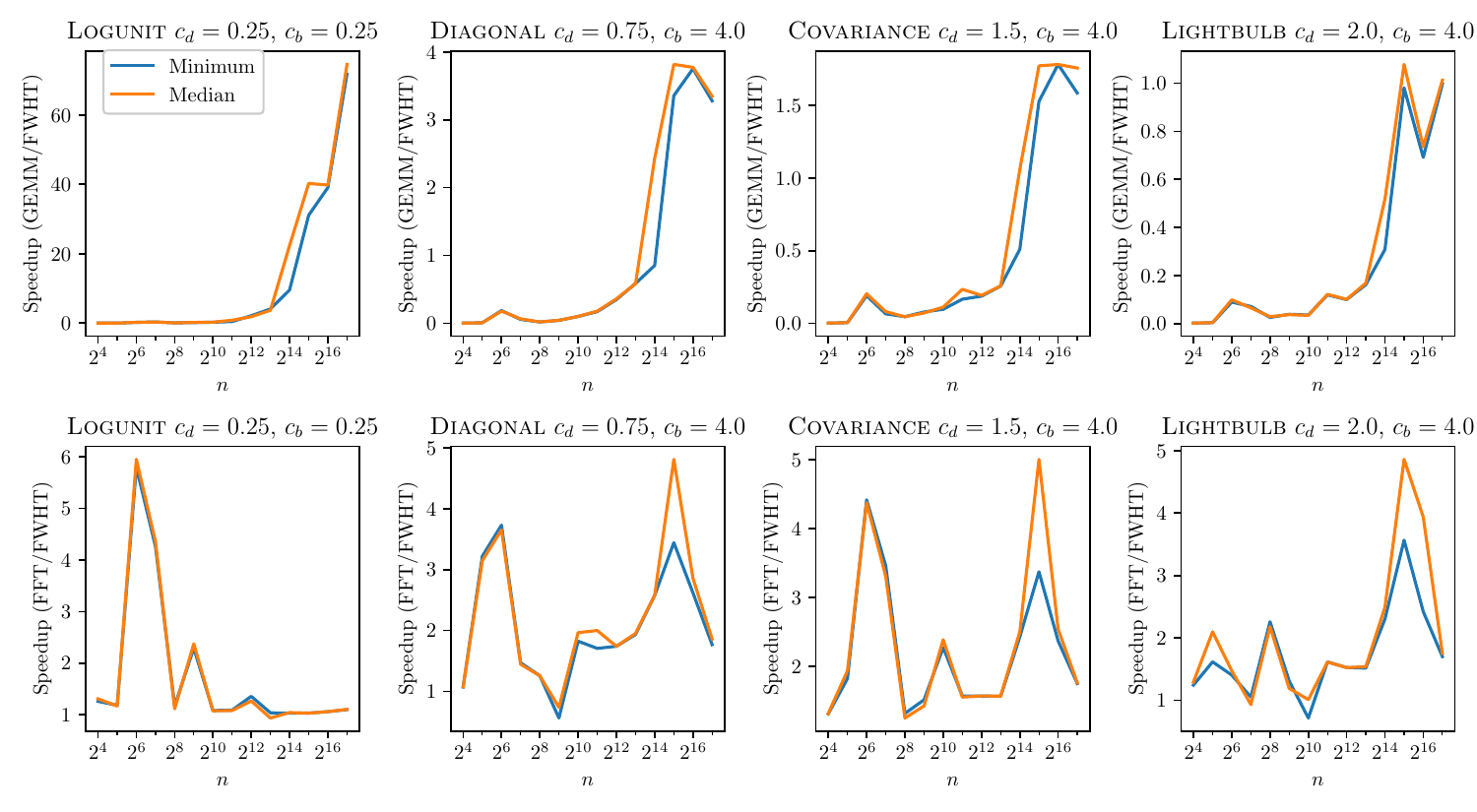}
    \caption{
        Speedup results from Vera. 
    }
    \label{fig:veraspeedup}
\end{figure}

\begin{figure}[h!]
    \centering
    \includegraphics[width=0.95\linewidth]{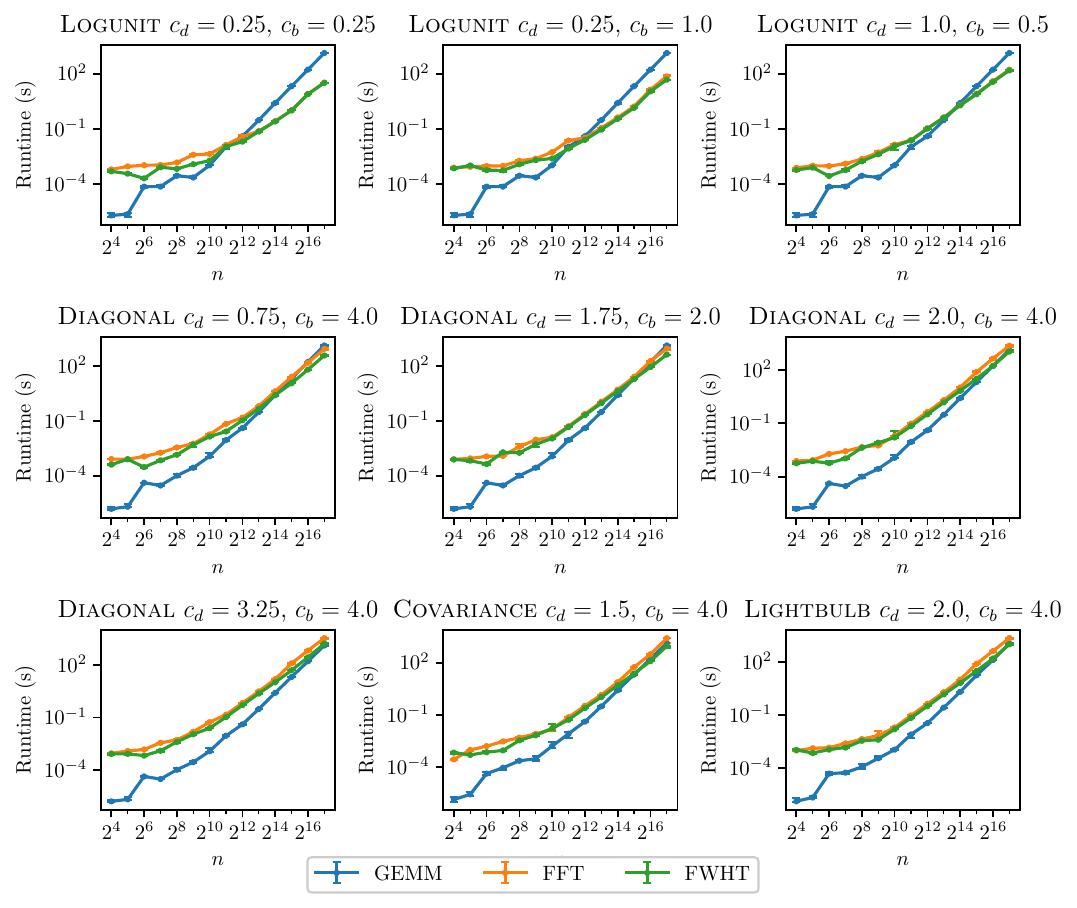}
    \caption{Scaling experiment results from Minerva.}
    \label{fig:minervascaling}
\end{figure}

\begin{figure}[h!]
    \centering
    \includegraphics[width=0.95\linewidth]{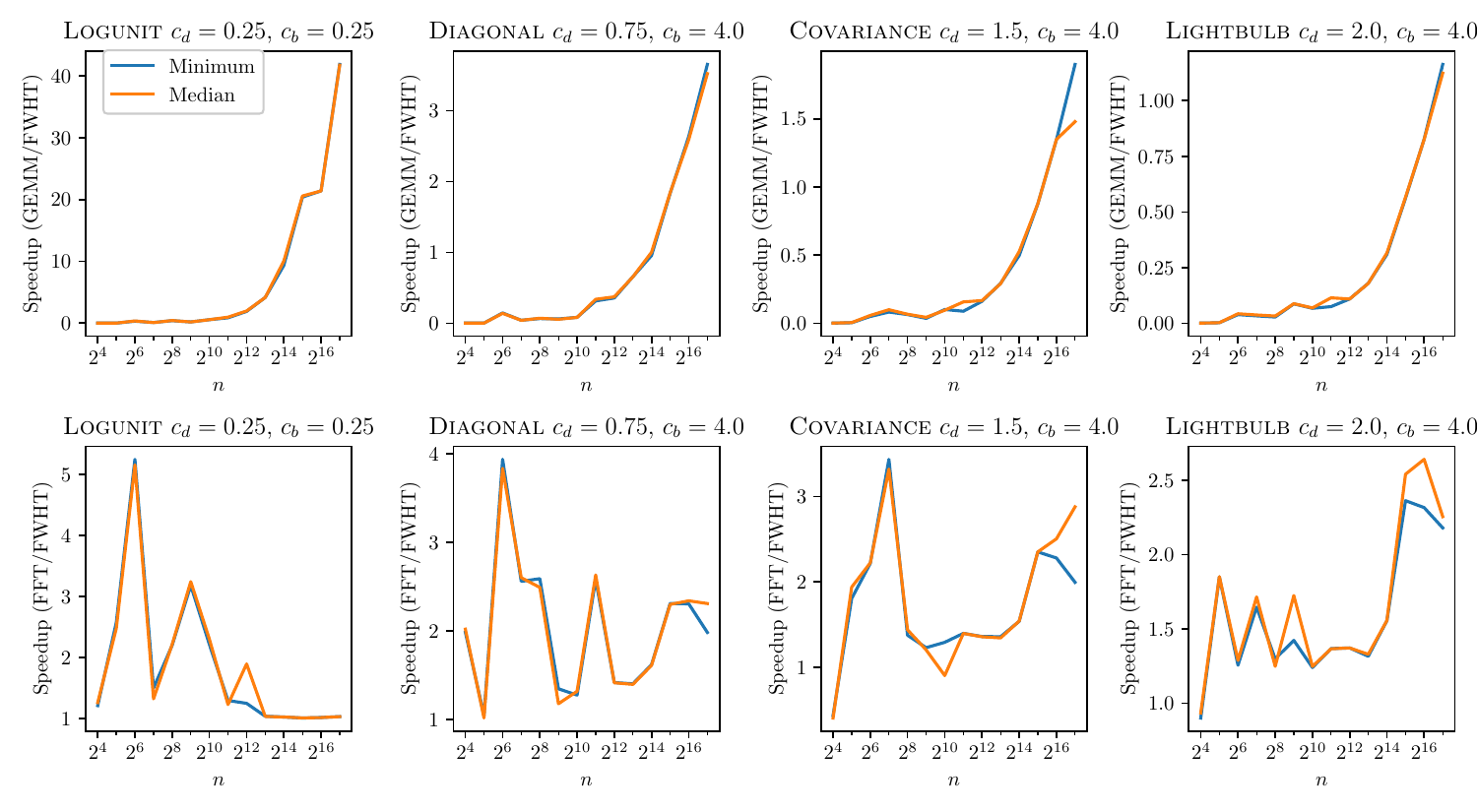}
    \caption{
        Speedup results from Minerva. 
    }
    \label{fig:minervaspeedup}
\end{figure}

\begin{table}[h]
    \centering
        \caption{Correctness results at $n=65536$ with $10\times 10$ repetitions. All $\epsilon 0.1$ means the average fraction
        of all elements within 0.1 in absolute error. Big $\epsilon \leq 0.1$ likewise for \emph{big} elements. Big $\geq 0.5$ and Small $\leq 0.5$ means the average fraction of all \emph{big} and \emph{small} elements with absolute value above or below 0.5.}
    \label{tab:correctness}

    \begin{tabular}{@{}lll@{}r@{\,}r@{\,}r@{\,}r@{}}
\textbf{Problem} & \textbf{Transform} & \textbf{Category} & \textbf{All $\epsilon \leq 0.1 $} & \textbf{Big $\epsilon \leq 0.1$} & \textbf{Big $\geq 0.5$} & \textbf{Small $\leq 0.5$}\\\hline
\textsc{Logunit} & FFT & Perfect & 100.00\% & 100.00\% & 100.00\% & 100.00\%\\
\textsc{Logunit} & FFT & Good & 100.00\% & 100.00\% & 100.00\% & 100.00\%\\
\textsc{Logunit} & FFT & Decent & 100.00\% & 100.00\% & 100.00\% & 100.00\%\\
\textsc{Logunit} & FFT & Satisfactory & 100.00\% & 100.00\% & 100.00\% & 100.00\%\\
\textsc{Diagonal} & FFT & Perfect & 100.00\% & 100.00\% & 100.00\% & 100.00\%\\
\textsc{Diagonal} & FFT & Good & 100.00\% & 100.00\% & 100.00\% & 100.00\%\\
\textsc{Diagonal} & FFT & Decent &  99.98\% &  99.98\% &  99.99\% &  99.99\%\\
\textsc{Diagonal} & FFT & Satisfactory &  99.97\% &  99.97\% &  99.98\% &  99.97\%\\
\textsc{Covariance} & FFT & Satisfactory &  57.96\% &  57.00\% & 100.00\% & 100.00\%\\
\textsc{Lightbulb} & FFT & Satisfactory &  64.41\% &  56.00\% & 100.00\% & 100.00\%\\
\textsc{Logunit} & FWHT & Perfect & 100.00\% & 100.00\% & 100.00\% & 100.00\%\\
\textsc{Logunit} & FWHT & Good & 100.00\% & 100.00\% & 100.00\% & 100.00\%\\
\textsc{Logunit} & FWHT & Decent & 100.00\% & 100.00\% & 100.00\% & 100.00\%\\
\textsc{Logunit} & FWHT & Satisfactory & 100.00\% & 100.00\% & 100.00\% & 100.00\%\\
\textsc{Diagonal} & FWHT & Perfect & 100.00\% & 100.00\% & 100.00\% & 100.00\%\\
\textsc{Diagonal} & FWHT & Good & 100.00\% & 100.00\% & 100.00\% & 100.00\%\\
\textsc{Diagonal} & FWHT & Decent &  99.98\% &  99.99\% &  99.99\% &  99.99\%\\
\textsc{Diagonal} & FWHT & Satisfactory &  99.97\% &  99.97\% &  99.99\% &  99.97\%\\
\textsc{Covariance} & FWHT & Satisfactory &  57.95\% &  60.00\% & 100.00\% & 100.00\%\\
\textsc{Lightbulb} & FWHT & Satisfactory &  64.42\% &  62.00\% &  99.00\% & 100.00\%\\
\end{tabular}
\end{table}

\begin{table}[h]
    \centering
    \caption{Running times of GEMM multiplication for the largest instances on Vera and Minerva. The times given are
    given here in seconds.}
    \label{tab:gemmruntimes}
    \begin{tabular}{@{}l@{\,\,}r@{\,\,}r@{\,\,}r@{\,\,}r@{\,\,}r@{\,\,}r@{\,\,}rr@{}}
 &  & \multicolumn{3}{c}{Vera} & \multicolumn{3}{c}{Minerva}\\
\textbf{Problem} & $n$ & Min & Med & Max & Min & Med & Max\\
\hline\textsc{Logunit} & $2^{15}$ & 36.3 & 47.1 & 53.5 & 20.9 & 21.1 & 21.1\\
\textsc{Logunit} & $2^{16}$ & 334.4 & 341.3 & 344.6 & 167.0 & 167.5 & 167.5\\
\textsc{Logunit} & $2^{17}$ & 2624.3 & 2726.6 & 2781.3 & 1344.8 & 1344.9 & 1345.0\\
\hline\textsc{Diagonal} & $2^{15}$ & 45.0 & 51.3 & 52.2 & 20.8 & 21.0 & 21.0\\
\textsc{Diagonal} & $2^{16}$ & 344.8 & 349.0 & 370.4 & 167.3 & 167.4 & 173.2\\
\textsc{Diagonal} & $2^{17}$ & 1950.4 & 2011.5 & 2208.6 & 1339.6 & 1340.6 & 1340.6\\
\hline\textsc{Covariance} & $2^{15}$ & 41.2 & 47.8 & 49.9 & 20.8 & 20.9 & 20.9\\
\textsc{Covariance} & $2^{16}$ & 339.9 & 345.6 & 352.3 & 166.4 & 167.3 & 168.1\\
\textsc{Covariance} & $2^{17}$ & 1911.9 & 2145.5 & 2169.2 & 1342.2 & 1342.4 & 1343.0\\
\hline\textsc{Lightbulb} & $2^{15}$ & 35.2 & 39.5 & 42.0 & 17.9 & 18.1 & 18.1\\
\textsc{Lightbulb} & $2^{16}$ & 168.0 & 179.4 & 200.0 & 132.9 & 133.1 & 133.3\\
\textsc{Lightbulb} & $2^{17}$ & 1724.6 & 1767.5 & 1887.5 & 1151.2 & 1151.5 & 1151.8\\
\end{tabular}
\end{table}

\begin{landscape}
\begin{table}[h]
    \centering
    \caption{Running times for largest instances on Minerva, given in seconds. SU1 = GEMM/FFT speedup,
    SU2 = GEMM/FWHT speedup, and SU3 = FFT/FWHT speedup.}
    \label{tab:verascalingtabular}
    \begin{tabular}{@{}l@{\,\,}r@{\,\,}r@{\,\,}r@{\,\,}r@{\,\,}r@{\,\,}r@{\,\,}r@{\,\,}r@{\,\,}r@{\,\,}r@{\,\,}r@{\,\,}r@{\,\,}r@{\,\,}r@{\,\,}r@{}}
 &  &  &  & \multicolumn{3}{c}{FFT} & \multicolumn{3}{c}{FWHT} & \multicolumn{2}{c}{SU1} & \multicolumn{2}{c}{SU2} & \multicolumn{2}{c}{SU3}\\
\textbf{Problem} & $n$ & $c_d$ & $c_b$ & Min & Med & Max & Min & Med & Max & Min & Med & Min & Med & Min & Med\\
\hline\textsc{Logunit} & $2^{15}$ & 0.25 & 0.25 & 1.2 & 1.2 & 1.2 & 1.2 & 1.2 & 1.2 & 30.22 & 39.16 & 31.10 & 40.27 & 1.03 & 1.03\\
\textsc{Logunit} & $2^{16}$ & 0.25 & 0.25 & 9.0 & 9.1 & 9.1 & 8.6 & 8.6 & 8.6 & 36.98 & 37.65 & 39.07 & 39.80 & 1.06 & 1.06\\
\textsc{Logunit} & $2^{17}$ & 0.25 & 0.25 & 40.1 & 40.2 & 40.3 & 36.6 & 36.6 & 36.6 & 65.43 & 67.77 & 71.79 & 74.57 & 1.10 & 1.10\\
\textsc{Logunit} & $2^{15}$ & 0.25 & 1.00 & 2.3 & 2.3 & 2.3 & 1.7 & 1.7 & 1.7 & 15.96 & 20.61 & 20.92 & 27.13 & 1.31 & 1.32\\
\textsc{Logunit} & $2^{16}$ & 0.25 & 1.00 & 22.6 & 22.8 & 22.9 & 12.7 & 12.7 & 13.3 & 14.78 & 14.96 & 26.31 & 26.79 & 1.78 & 1.79\\
\textsc{Logunit} & $2^{17}$ & 0.25 & 1.00 & 107.9 & 120.2 & 120.5 & 55.2 & 55.5 & 55.5 & 24.31 & 22.68 & 47.56 & 49.16 & 1.96 & 2.17\\
\textsc{Logunit} & $2^{15}$ & 1.00 & 0.50 & 9.6 & 9.6 & 9.6 & 8.7 & 8.7 & 8.7 & 3.78 & 4.90 & 4.20 & 5.44 & 1.11 & 1.11\\
\textsc{Logunit} & $2^{16}$ & 1.00 & 0.50 & 45.5 & 45.5 & 45.6 & 40.3 & 40.3 & 40.3 & 7.35 & 7.50 & 8.31 & 8.47 & 1.13 & 1.13\\
\textsc{Logunit} & $2^{17}$ & 1.00 & 0.50 & 222.1 & 225.8 & 226.2 & 164.0 & 164.1 & 164.2 & 11.81 & 12.08 & 16.00 & 16.61 & 1.35 & 1.38\\
\hline\textsc{Diagonal} & $2^{15}$ & 0.75 & 4.00 & 46.2 & 64.7 & 65.4 & 13.4 & 13.4 & 13.5 & 0.97 & 0.79 & 3.36 & 3.82 & 3.44 & 4.81\\
\textsc{Diagonal} & $2^{16}$ & 0.75 & 4.00 & 240.5 & 264.3 & 264.5 & 91.9 & 92.5 & 92.8 & 1.43 & 1.32 & 3.75 & 3.77 & 2.62 & 2.86\\
\textsc{Diagonal} & $2^{17}$ & 0.75 & 4.00 & 1051.4 & 1117.4 & 1126.0 & 594.8 & 600.7 & 603.8 & 1.85 & 1.80 & 3.28 & 3.35 & 1.77 & 1.86\\
\textsc{Diagonal} & $2^{15}$ & 1.75 & 2.00 & 41.7 & 43.3 & 44.7 & 21.4 & 21.6 & 21.7 & 1.08 & 1.19 & 2.10 & 2.38 & 1.95 & 2.01\\
\textsc{Diagonal} & $2^{16}$ & 1.75 & 2.00 & 265.7 & 320.5 & 323.6 & 100.3 & 100.3 & 100.3 & 1.30 & 1.09 & 3.44 & 3.48 & 2.65 & 3.19\\
\textsc{Diagonal} & $2^{17}$ & 1.75 & 2.00 & 1209.0 & 1317.8 & 1588.7 & 621.4 & 622.0 & 622.2 & 1.61 & 1.53 & 3.14 & 3.23 & 1.95 & 2.12\\
\textsc{Diagonal} & $2^{15}$ & 2.00 & 4.00 & 128.0 & 181.0 & 181.2 & 36.5 & 36.7 & 36.9 & 0.35 & 0.28 & 1.23 & 1.40 & 3.50 & 4.93\\
\textsc{Diagonal} & $2^{16}$ & 2.00 & 4.00 & 576.2 & 771.1 & 775.3 & 259.9 & 263.4 & 264.9 & 0.60 & 0.45 & 1.33 & 1.33 & 2.22 & 2.93\\
\textsc{Diagonal} & $2^{17}$ & 2.00 & 4.00 & 3029.7 & 3168.0 & 3189.8 & 3098.0 & 3133.1 & 3135.9 & 0.64 & 0.63 & 0.63 & 0.64 & 0.98 & 1.01\\
\textsc{Diagonal} & $2^{15}$ & 3.25 & 4.00 & 205.3 & 254.7 & 257.5 & 57.2 & 57.7 & 59.3 & 0.22 & 0.20 & 0.79 & 0.89 & 3.59 & 4.42\\
\textsc{Diagonal} & $2^{16}$ & 3.25 & 4.00 & 1002.5 & 1153.3 & 1190.6 & 408.2 & 409.3 & 441.5 & 0.34 & 0.30 & 0.84 & 0.85 & 2.46 & 2.82\\
\textsc{Diagonal} & $2^{17}$ & 3.25 & 4.00 & 4697.2 & 4787.1 & 4811.3 & 2743.8 & 2804.9 & 2879.7 & 0.42 & 0.42 & 0.71 & 0.72 & 1.71 & 1.71\\
\hline\textsc{Covariance} & $2^{15}$ & 1.50 & 4.00 & 90.8 & 135.0 & 135.0 & 26.9 & 27.0 & 27.1 & 0.45 & 0.35 & 1.53 & 1.77 & 3.37 & 5.01\\
\textsc{Covariance} & $2^{16}$ & 1.50 & 4.00 & 451.7 & 492.6 & 494.6 & 190.7 & 193.9 & 194.8 & 0.75 & 0.70 & 1.78 & 1.78 & 2.37 & 2.54\\
\textsc{Covariance} & $2^{17}$ & 1.50 & 4.00 & 2109.6 & 2149.5 & 2153.0 & 1205.7 & 1220.8 & 1236.3 & 0.91 & 1.00 & 1.59 & 1.76 & 1.75 & 1.76\\
\hline\textsc{Lightbulb} & $2^{15}$ & 2.00 & 4.00 & 127.9 & 178.2 & 178.7 & 35.9 & 36.6 & 36.7 & 0.28 & 0.22 & 0.98 & 1.08 & 3.57 & 4.86\\
\textsc{Lightbulb} & $2^{16}$ & 2.00 & 4.00 & 588.3 & 960.5 & 961.6 & 243.0 & 243.9 & 244.6 & 0.29 & 0.19 & 0.69 & 0.74 & 2.42 & 3.94\\
\textsc{Lightbulb} & $2^{17}$ & 2.00 & 4.00 & 2939.8 & 3055.9 & 3138.5 & 1730.2 & 1745.5 & 1748.7 & 0.59 & 0.58 & 1.00 & 1.01 & 1.70 & 1.75\\
\end{tabular}
\end{table}

\begin{table}[h]
    \centering
    \caption{Running times for largest instances on Minerva, given in seconds. SU1 = GEMM/FFT speedup,
    SU2 = GEMM/FWHT speedup, and SU3 = FFT/FWHT speedup.}
    \label{tab:minervascalingtabular}
    \begin{tabular}{@{}l@{\,\,}r@{\,\,}r@{\,\,}r@{\,\,}r@{\,\,}r@{\,\,}r@{\,\,}r@{\,\,}r@{\,\,}r@{\,\,}r@{\,\,}r@{\,\,}r@{\,\,}r@{\,\,}r@{\,\,}r@{}}
 &  &  &  & \multicolumn{3}{c}{FFT} & \multicolumn{3}{c}{FWHT} & \multicolumn{2}{c}{SU1} & \multicolumn{2}{c}{SU2} & \multicolumn{2}{c}{SU3}\\
\textbf{Problem} & $n$ & $c_d$ & $c_b$ & Min & Med & Max & Min & Med & Max & Min & Med & Min & Med & Min & Med\\
\hline\textsc{Logunit} & $2^{15}$ & 0.25 & 0.25 & 1.0 & 1.0 & 1.0 & 1.0 & 1.0 & 1.0 & 20.15 & 20.29 & 20.39 & 20.58 & 1.01 & 1.01\\
\textsc{Logunit} & $2^{16}$ & 0.25 & 0.25 & 8.0 & 8.0 & 8.0 & 7.8 & 7.8 & 7.8 & 20.96 & 20.98 & 21.40 & 21.39 & 1.02 & 1.02\\
\textsc{Logunit} & $2^{17}$ & 0.25 & 0.25 & 33.3 & 33.3 & 33.4 & 32.1 & 32.2 & 32.3 & 40.44 & 40.34 & 41.87 & 41.76 & 1.04 & 1.04\\
\textsc{Logunit} & $2^{15}$ & 0.25 & 1.00 & 1.6 & 1.7 & 1.7 & 1.4 & 1.4 & 1.4 & 12.68 & 12.58 & 15.04 & 14.96 & 1.19 & 1.19\\
\textsc{Logunit} & $2^{16}$ & 0.25 & 1.00 & 13.8 & 13.9 & 14.0 & 10.7 & 11.0 & 11.1 & 12.13 & 12.07 & 15.58 & 15.22 & 1.29 & 1.26\\
\textsc{Logunit} & $2^{17}$ & 0.25 & 1.00 & 71.9 & 72.9 & 88.2 & 47.8 & 48.0 & 48.3 & 18.70 & 18.46 & 28.11 & 28.01 & 1.50 & 1.52\\
\textsc{Logunit} & $2^{15}$ & 1.00 & 0.50 & 8.3 & 8.3 & 8.3 & 7.7 & 7.7 & 7.8 & 2.53 & 2.55 & 2.71 & 2.73 & 1.07 & 1.07\\
\textsc{Logunit} & $2^{16}$ & 1.00 & 0.50 & 39.3 & 39.3 & 39.4 & 36.6 & 36.6 & 36.7 & 4.25 & 4.26 & 4.56 & 4.57 & 1.07 & 1.07\\
\textsc{Logunit} & $2^{17}$ & 1.00 & 0.50 & 167.4 & 167.6 & 167.8 & 150.8 & 151.7 & 151.8 & 8.03 & 8.03 & 8.92 & 8.86 & 1.11 & 1.10\\
\hline\textsc{Diagonal} & $2^{15}$ & 0.75 & 4.00 & 26.2 & 26.2 & 26.3 & 11.3 & 11.4 & 11.4 & 0.79 & 0.80 & 1.83 & 1.84 & 2.31 & 2.30\\
\textsc{Diagonal} & $2^{16}$ & 0.75 & 4.00 & 146.3 & 151.1 & 163.1 & 63.5 & 64.6 & 64.8 & 1.14 & 1.11 & 2.64 & 2.59 & 2.30 & 2.34\\
\textsc{Diagonal} & $2^{17}$ & 0.75 & 4.00 & 726.9 & 878.7 & 1003.6 & 366.9 & 380.8 & 380.9 & 1.84 & 1.53 & 3.65 & 3.52 & 1.98 & 2.31\\
\textsc{Diagonal} & $2^{15}$ & 1.75 & 2.00 & 25.9 & 26.0 & 26.0 & 19.2 & 19.2 & 19.4 & 0.80 & 0.81 & 1.08 & 1.09 & 1.35 & 1.35\\
\textsc{Diagonal} & $2^{16}$ & 1.75 & 2.00 & 168.2 & 189.4 & 190.1 & 88.2 & 88.4 & 88.5 & 0.99 & 0.88 & 1.90 & 1.89 & 1.91 & 2.14\\
\textsc{Diagonal} & $2^{17}$ & 1.75 & 2.00 & 801.7 & 910.6 & 933.3 & 424.1 & 425.5 & 428.1 & 1.67 & 1.47 & 3.16 & 3.15 & 1.89 & 2.14\\
\textsc{Diagonal} & $2^{15}$ & 2.00 & 4.00 & 75.0 & 77.1 & 82.5 & 32.1 & 32.1 & 32.2 & 0.28 & 0.27 & 0.65 & 0.65 & 2.34 & 2.40\\
\textsc{Diagonal} & $2^{16}$ & 2.00 & 4.00 & 370.5 & 436.1 & 452.0 & 167.9 & 169.9 & 170.7 & 0.45 & 0.38 & 1.00 & 0.99 & 2.21 & 2.57\\
\textsc{Diagonal} & $2^{17}$ & 2.00 & 4.00 & 2071.5 & 2301.3 & 2377.8 & 1029.7 & 1086.4 & 1106.8 & 0.65 & 0.58 & 1.30 & 1.23 & 2.01 & 2.12\\
\textsc{Diagonal} & $2^{15}$ & 3.25 & 4.00 & 117.4 & 127.2 & 129.8 & 50.9 & 50.9 & 50.9 & 0.18 & 0.17 & 0.41 & 0.41 & 2.31 & 2.50\\
\textsc{Diagonal} & $2^{16}$ & 3.25 & 4.00 & 619.1 & 684.2 & 742.4 & 274.9 & 282.4 & 283.3 & 0.27 & 0.24 & 0.61 & 0.59 & 2.25 & 2.42\\
\textsc{Diagonal} & $2^{17}$ & 3.25 & 4.00 & 3334.8 & 3539.1 & 3612.9 & 1625.3 & 1713.4 & 1746.9 & 0.40 & 0.38 & 0.82 & 0.78 & 2.05 & 2.07\\
\hline\textsc{Covariance} & $2^{15}$ & 1.50 & 4.00 & 55.8 & 56.0 & 58.5 & 23.7 & 23.8 & 23.8 & 0.37 & 0.37 & 0.88 & 0.88 & 2.35 & 2.35\\
\textsc{Covariance} & $2^{16}$ & 1.50 & 4.00 & 281.7 & 310.5 & 319.3 & 123.5 & 124.0 & 124.1 & 0.59 & 0.54 & 1.35 & 1.35 & 2.28 & 2.50\\
\textsc{Covariance} & $2^{17}$ & 1.50 & 4.00 & 1410.6 & 2612.0 & 2682.2 & 706.5 & 907.3 & 926.6 & 0.95 & 0.51 & 1.90 & 1.48 & 2.00 & 2.88\\
\hline\textsc{Lightbulb} & $2^{15}$ & 2.00 & 4.00 & 75.4 & 81.2 & 81.7 & 31.9 & 32.0 & 32.0 & 0.24 & 0.22 & 0.56 & 0.57 & 2.36 & 2.54\\
\textsc{Lightbulb} & $2^{16}$ & 2.00 & 4.00 & 372.2 & 425.3 & 430.4 & 160.8 & 161.1 & 161.3 & 0.36 & 0.31 & 0.83 & 0.83 & 2.32 & 2.64\\
\textsc{Lightbulb} & $2^{17}$ & 2.00 & 4.00 & 2157.5 & 2311.2 & 2344.0 & 990.2 & 1025.0 & 1035.7 & 0.53 & 0.50 & 1.16 & 1.12 & 2.18 & 2.25\\
\end{tabular}
\end{table}
\end{landscape}

\end{document}